\documentclass[reqno]{amsart}
\usepackage{amssymb}

\newtheorem{theorem}{Theorem}[section]
\newtheorem{proposition}[theorem]{Proposition}

\newtheorem{corollary}[theorem]{Corollary}

\theoremstyle{remark}
\newtheorem{remark}[theorem]{Remark}

\numberwithin{equation}{section}

\begin{document}

\title[infinite $q$-boson system]
{Diagonalization of the infinite $q$-boson system}

\author{J.F.  van Diejen}

\author{E. Emsiz}

\address{
Facultad de Matem\'aticas, Pontificia Universidad Cat\'olica de Chile,
Casilla 306, Correo 22, Santiago, Chile}
\email{diejen@mat.puc.cl, eemsiz@mat.puc.cl}


\thanks{Work was supported in part by the {\em Fondo Nacional de Desarrollo
Cient\'{\i}fico y Tecnol\'ogico (FONDECYT)} Grants \# 1130226 and  \# 11100315,
and by the {\em Anillo ACT56 `Reticulados y Simetr\'{\i}as'}
financed by the  {\em Comisi\'on Nacional de Investigaci\'on
Cient\'{\i}fica y Tecnol\'ogica (CONICYT)}}

\date{March 2013}

\begin{abstract}
We present a hierarchy of commuting operators in Fock space containing the $q$-boson Hamiltonian on $\mathbb{Z}$ and show that the operators in question are simultaneously diagonalized by Hall-Littlewood functions. As an application, the $n$-particle scattering operator is computed.
\end{abstract}

\maketitle

\section{Introduction}\label{sec1}
The $q$-boson model constitutes a one-dimensional exactly solvable particle system in Fock space \cite{bog-ize-kit:correlation} based on the $q$-oscillator algebra \cite[\text{Ch.}~5]{kli-smu:quantum}.
In the case of periodic boundary conditions (i.e. with particles hopping on the finite lattice $\mathbb{Z}_m$), the integrability, the spectrum, and the eigenfunctions of the Hamiltonian were analyzed by means of the algebraic Bethe Ansatz method \cite{bog-ize-kit:correlation}. Remarkably, these eigenfunctions turn out to be Hall-Littlewood functions \cite{tsi:quantum,kor:cylindric} (cf. also \cite{jin:vertex} for an alternative construction of Hall-Littlewood functions in Fock space based on deformed vertex operator algebras, with applications in the study of
KP $\tau$-functions arising from generating functions of weighted plane partitions
\cite{fod-whe:hall-littlewood}).
With the aid of explicit expressions for the commuting quantum integrals arising from an infinite-dimensional solution
of the Yang-Baxter equation, it was very recently demonstrated \cite{kor:cylindric} that the eigenvalue problem for the $q$-boson system on $\mathbb{Z}_m$ is in fact equivalent to that of an integrable discretization \cite{die:diagonalization} of the celebrated delta Bose gas on the circle \cite{lie-lin:exact}.

The present work addresses the spectral problem and the integrability of the $q$-boson system on the {\em infinite} lattice $\mathbb{Z}$. Specifically, we demonstrate that the eigenfunctions of this infinite $q$-boson system are again given by Hall-Littlewood functions and provide explicit formulas for a complete hierarchy of operators commuting with the Hamiltonian; these formulas are natural infinite-dimensional analogues of the above-mentioned expressions in \cite{kor:cylindric} for the finite $q$-boson system on $\mathbb{Z}_m$.
Finally, the $n$-particle scattering operator is computed as an application of Ruijsenaars' general scattering results in \cite{rui:factorized}.

\section{The infinite $q$-boson system}
Given $n\geq 0$ integral, let
$\mathcal{F}(\Lambda_n)$  be the space of complex functions  $f:\Lambda_n\to\mathbb{C}$ on the fundamental domain
\begin{equation}\label{dominant}
\Lambda_n:=\{(\lambda_1,\dots, \lambda_n)\in\mathbb{Z}^n \mid \lambda_1\geq \lambda_2 \geq \cdots \geq \lambda_n\}
\end{equation}
of the integral lattice
$\mathbb{Z}^n$ modulo the action of the permutation group $S_n$,
where $\Lambda_0:=\{ 0\}$ and $\mathcal{F}(\Lambda_0):=\mathbb{C}$ by convention.
We will refer to the infinite direct sum
\begin{equation}\label{AFock}
\mathcal{F}:=\bigoplus_{n\geq 0} \mathcal{F}(\Lambda_n)
\end{equation}
built of all {\em finite} linear combinations of  (arbitrary) functions $f_n\in\mathcal{F}(\Lambda_n)$, $n=0,1,2,\ldots $ as the algebraic Fock space.

For $\lambda\in \Lambda_n$ and $l\in \mathbb{Z}$, let the multiplicity $m_l(\lambda)$ count the number of components
$\lambda_j$, $1\leq j\leq n$ such that $\lambda_j=l$.
 We write $\beta^*_l\lambda$ for the point of $\Lambda_{n+1}$ obtained from $\lambda$ by inserting an additional component with value $l$
and---assuming   $m_l(\lambda)>0$---we write $\beta_l\lambda\in \Lambda_{n-1}$ for the result of the inverse operation that deletes  a component with value $l$ from $\lambda$.
Upon defining the following actions on $f\in\mathcal{F}(\Lambda_n)$:
\begin{equation}\label{qboson-rep}
    \begin{split}
  (\beta_l f)(\lambda)&:=
  \begin{cases} f(\beta_l^*\lambda) &\text{if}\  n>0\ (\lambda\in \Lambda_{n-1}),\\
   0& \text{if}\  n=0
   \end{cases}
  \\
(\beta^*_l f)(\lambda)&:=
\begin{cases}
[m_l(\lambda)]f(\beta_l\lambda)&\text{if}\  m_l(\lambda)>0 \\
0&\text{otherwise}
\end{cases}
 \quad ( \lambda\in \Lambda_{n+1}),\\
(N_l f)(\lambda)&:=q^{m_l(\lambda)}f(\lambda) \qquad\quad (\lambda\in \Lambda_n),
    \end{split}
\end{equation}
where $0<q<1$ and
$$[m]:=\frac{1-q^{m}}{1-q}=1+q+\cdots +q^{m-1}$$
for $m=0,1,2,\ldots$, it is readily verified that one ends up with a representation of the $q$-boson field algebra on $\mathcal{F}$:
\begin{subequations}
\begin{equation}
[\beta_l,\beta_k]=[\beta^*_l,\beta^*_k]=[N_l,N_k]=[N_l,\beta_k]=[N_l,\beta^*_k]=[\beta_l,\beta^*_k]= 0
\end{equation}
for $l\neq k$, and
\begin{equation}
N_l\beta_l^*=q \beta_l^* N_l,
\  \beta_l N_l = q  N_l\beta_l,
\ [\beta_l,\beta_l^*]=N_l,
\ \beta_l\beta_l^* - q\beta_l^*\beta_l=1.
\end{equation}
\end{subequations}
Here the brackets refer to the (ordinary) commutator product.
By construction $\beta_l$, $\beta_l^*$  and $N_l$ map $\mathcal{F}(\Lambda_n)$ into $\mathcal{F}(\Lambda_{n-1})$, $\mathcal{F}(\Lambda_{n+1})$ and $\mathcal{F}(\Lambda_n)$, respectively (with the convention that $\mathcal{F}(\Lambda_{-1})=\{0\}$).

The Hamiltonian of the $q$-boson system
\begin{subequations}
\begin{equation}\label{qbH}
\text{H}_q = \sum_{l\in \mathbb{Z}} (a_l + a_l^*)
\end{equation}
is built of hopping operators
\begin{equation}\label{hopping-op}
a_l:=\beta_{l+1}^*\beta_l  \quad\text{and}\quad  a^*_l:=\beta_{l+1}\beta^*_l
\end{equation}
\end{subequations}
for which the $n$-particle subspace $\mathcal{F}(\Lambda_n)$ is stable. These hopping operators
represent the plactic subalgebra of the $q$-boson field algebra \cite[\text{Sec.}~3.4]{kor:cylindric}:
\begin{subequations}
\begin{equation}
a_l a_k = a_k a_l
\end{equation}
for $|l-k|>1$ (nonlocal commutativity) and
\begin{equation}
  \begin{split}
a_{l+1}a_l^2 + q a_l^2 a_{l+1} &= (1+q) a_l a_{l+1} a_l\\
a_{l+1}^2a_l + qa_l a_{l+1}^2 &= (1+q) a_{l+1} a_j a_{l+1}
      \end{split}
 \end{equation}
 \end{subequations}
(quantum Knuth relations), with analogous relations (involving reversely ordered products) for $a_l^*$, $l\in\mathbb{Z}$.
The  $q$-boson Hamiltonian $\text{H}_q$ \eqref{qbH}, \eqref{hopping-op} constitutes a well-defined operator on $\mathcal{F}$ as for any $f\in \mathcal{F}(\Lambda_n)$ and $\lambda\in\Lambda_n$ the infinite sum $(\text{H}_q f)(\lambda)$ containes only a {\em finite} number of nonvanishing terms.

To facilitate the comparison with previous literature on the $q$-boson system \cite{bog-ize-kit:correlation,tsi:quantum,kor:cylindric}, let us denote the characteristic function in $\mathcal{F}(\Lambda_n)$ supported on $\lambda\in\Lambda_n$ by $|\lambda\rangle$. Then one has that
\begin{equation*}
\beta_l |\lambda\rangle = \begin{cases}
|\beta_l\lambda\rangle &\text{if}\ m_l(\lambda)>0 \\
0&\text{otherwise}
\end{cases},\quad
\beta_l^* |\lambda\rangle =[m_l(\lambda)+1] | \beta_l^*\lambda\rangle,\quad N_l|\lambda\rangle=q^{m_l(\lambda)}|\lambda\rangle .
\end{equation*}
In the standard physical interpretation the state $|\lambda\rangle$ encodes a configuration of $n$ particles on $\mathbb{Z}$---$q$-bosons---with $m_l(\lambda)$ particles occupying the site $l\in\mathbb{Z}$.  The operators $\beta_l^*$ and $\beta_l$ play the role of particle creation and annihilation operators and $N_l$ counts the number of particles at the site $l$ (as a power of $q$). The hopping operators $a_l$ and $a_l^*$ move a particle from $l$ to $l+1$ and vice versa.

\section{Integrability}
To any partition $\eta=(\eta_1,\ldots,\eta_p )$ with $\eta_1\geq\eta_2\geq  \cdots \geq \eta_p\geq 1$, we associate
the following hopping operators on $\mathcal{F}$:
\begin{equation}\label{monomial}
  \begin{split}
    m_\eta(a)&:=\sideset{}{'}\sum_{\sigma\in S_p}  \sum_{l_1 < l_2 < \cdots < l_p} a_{l_1}^{\eta_{\sigma_1}} \cdots a_{l_p}^{\eta_{\sigma_p}}, \\
  m_\eta(a^*)&:=\sideset{}{'}\sum_{\sigma\in S_p} \sum_{l_1 > l_2 > \cdots > l_p} (a^*_{l_1})^{\eta_{\sigma_1}} \cdots (a^*_{l_p})^{\eta_{\sigma_p}}.
  \end{split}
  \end{equation}
 Here the (infinite) inner summations are over all strictly monotonous $p$-tuples $(l_1,\ldots ,l_p)$ of indices in $\mathbb{Z}$; the primes attached to the (finite) outer summations indicate that these are meant over the orbit of all {\em distinct} compositions $(\eta_{\sigma_1},\ldots ,\eta_{\sigma_p})$ obtained by reordering the parts of $\eta$ via permutations $$\sigma=\left( \begin{matrix} 1& 2& \cdots & p \\
 \sigma_1&\sigma_2&\cdots & \sigma_p
 \end{matrix}\right)$$ belonging to the symmetric group $S_p$. Notice that for given $f\in\mathcal{F}(\Lambda_n)$ and $\lambda\in\Lambda_n$, the
infinite sums $(m_\eta(a)f)(\lambda)$ and
 $(m_\eta(a^*)f)(\lambda)$ contain only a finite number of nonzero terms, so these operators are again well-defined on
 $\mathcal{F}$.
 For $r\in\mathbb{N}$ we now set
\begin{equation}\label{integrals}
  \begin{split}
H_r&:=\sum_{|\eta|=r} \frac{m_\eta(a)}{[\eta]!} \quad \text{and}\quad H^*_r:=\sum_{|\eta|=r} \frac{m_\eta(a^*)}{[\eta]!} ,
  \end{split}
\end{equation}
where $[\eta]!=[(\eta_1,\dots,\eta_p)]! := [\eta_1]! \cdots [\eta_p]!$ and $[m]!:=[m]  [m-1]\cdots [1]$ for $m=0,1,2,\ldots$ (with the convention that $[0]!=1$), and  $|\eta |:=\eta_1+\cdots +\eta_p$ (so
the (finite) summation in Eq. \eqref{integrals} is over all partitions of $r$). The $q$-boson Hamiltonian \eqref{qbH}, \eqref{hopping-op} becomes in terms of these operators:
\begin{equation}\label{hamiltonian}
\text{H}_q=H_1+H_1^*.
\end{equation}

Our main result is the following explicit formula for the action of $H_r$ and $H_r^*$ in the $n$-particle subspace $\mathcal{F}(\Lambda_n)$, which will be proven shortly in the next section.

\begin{theorem}[Explicit action of $H_r^{(*)}$ in $\mathcal{F}(\Lambda_n)$] \label{Hr-action:thm}
For any $f\in \mathcal{F}(\Lambda_n)$ and $\lambda\in\Lambda_n$, one has that
\begin{subequations}
 \begin{equation}\label{Hr-action}
    \begin{split}
(H_r f)(\lambda) &=  \sum_{\substack{J\subset\{1,2,\ldots,n\} ,|J|=r\\ \lambda-e_J\in\Lambda_n  }} V_{\lambda,J^c}\, f(\lambda-e_J), \\
(H^*_rf)(\lambda) &=  \sum_{\substack{J\subset \{1,2,\ldots,n\} ,|J|=r\\ \lambda+e_J\in\Lambda_n }} V_{\lambda,J}\, f(\lambda+e_J) ,
    \end{split}
 \end{equation}
 where $|J|$ denotes the cardinality of $J\subset \{1,\ldots,n\}$, $J^c:=\{ 1,\ldots ,n\}\setminus J$, $e_J:=\sum_{j\in J} e_j$ (with $e_1,\ldots ,e_n$ referring to the standard unit basis of $\mathbb{Z}^n$) and
 \begin{equation}\label{V}
  V_{\lambda,J}:= \prod_{\substack{1 \leq j < k\leq n \\ j\in J, k\in J^c \\ \lambda_j=\lambda_k}}
  \frac{1-q^{k-j+1}}{1-q^{k-j}} .
\end{equation}
\end{subequations}

In particular, for $r>n$ the $n$-particle subspace $\mathcal{F}(\Lambda_n)$ belongs to the kernel of the operators $H_r$ and $H_r^*$.
 \end{theorem}
In \cite[App.~C]{die-ems:unitary} the discrete difference operators $H_r$ on the RHS of Eqs. \eqref{Hr-action}, \eqref{V} were obtained (up to a trivial similarity transformation and replacing $q$ by $q^2$) as $n$ algebraically independent commuting central elements arising from a difference-reflection representation of  the affine Hecke algebra associated with $GL(n;\mathbb{C})$.
Since $H_n$ acts on $f\in \mathcal{F}(\Lambda_n)$ simply as an overall translationonal symmetry: $(H_nf)(\lambda)=f(\lambda-(e_1+\cdots +e_n))$,
it is immediate from the above formulas that the discrete difference operators in $\mathcal{F}(\Lambda_n)$ corresponding to $H_r^*$ can be written  in turn as
$H_{n-r}H_n^{-1}$  (with the convention that $H_0:=1$). The upshot is that all $H_r$ and $H_{r}^*$ \eqref{integrals} commute as operators on $\mathcal{F}$, which proves the integrability of the infinite $q$-boson system.

\begin{corollary}[Integrability]\label{integrability:cor}
All operators $H_r$ and $H_{r}^*$ in Eq. \eqref{integrals} mutually commute on $\mathcal{F}$:
\begin{equation}
[H_r,H_{r^\prime}] =0,\quad [H_r^*,H_{r^\prime}^*] =0, \quad [H_r ,H_{r^\prime}^*] =0\qquad
(\forall r,r^\prime \in\mathbb{N}),
\end{equation}
and they restrict to $n$ algebraically independent operators on the invariant $n$-particle subspace $\mathcal{F}(\Lambda_n)$.
\end{corollary}

A quasi-periodic counterpart of the formula in Theorem \ref{Hr-action:thm} for the {\em finite} $q$-boson system on $\mathbb{Z}_m$ can be found in Ref.~\cite{kor:cylindric} (see Prp.~3.11 and Prp.~6.1).
For Dirichlet type boundary conditions corresponding to the case of a vanishing quasi-periodicity parameter,
one arrives---in the limit when the lattice size parameter $m$ tends to infinity---at an analogue of the commutativity in Corollary \ref{integrability:cor} for the $q$-boson system on the (semi-)infinite lattice $\mathbb{N}$
as a consequence of \cite[Cor.~\text{3.3}]{kor:cylindric} (cf. also \cite[\text{Thm}.~5.3]{die:diagonalization}).
In principle, the commutativity in Corollary \ref{integrability:cor} for the $q$-boson system on $\mathbb{Z}$ could also be recovered along these lines upon centering the finite lattice around the origin before performing the infinite size limit. Alternatively, the commutativity in question can also be viewed as a degeneration of the commutativity of the discrete Macdonald-Ruijsenaars operators \cite{rui:complete}, \cite[\text{Sec.}~VI.6]{mac:symmetric}
via a limit transition (that takes Macdonald symmetric functions to Hall-Littlewood symmetric functions).

It is straightforward from Theorem \ref{Hr-action:thm} and Eq. \eqref{hamiltonian} that the action of the $q$-boson Hamiltonian $\text{H}_q$ \eqref{qbH}, \eqref{hopping-op} in the $n$-particle subspace $\mathcal{F}(\Lambda_n)$ is given by
\begin{equation}
(\text{H}_q f)(\lambda) =\sum_{\substack{1\leq j\leq n,\, \epsilon=\pm 1 \\ \lambda+\epsilon e_j\in \Lambda_n}} [m_{\lambda_j}(\lambda)]f(\lambda+\epsilon e_j)
\end{equation}
($f\in \mathcal{F}(\Lambda_n)$, $\lambda\in\Lambda_n$).

\section{Proof of  the main Theorem \ref{Hr-action:thm}}
We will determine the action of $H_r^*$ on $f\in\mathcal{F}_n$ by direct computation in three steps; the calculation of $H_rf$ is completely analogous so its details will be suppressed (but cf. Remark \ref{Hr-action:rem} below for an alternative shortcut yielding the action of $H_r$ from that of $H_r^*$ via adjointness).

\subsection{}\label{step1}
In the first step $V_{\lambda ,J}$ \eqref{V} (with $\lambda, \lambda+e_J\in\Lambda_n$) is rewritten in $q$-binomial form by means of well-known product formulas for the Poincar\'e polynomial
of the symmetric group:
\begin{subequations}
\begin{equation}\label{poincare}
S_n(q):=\sum_{\sigma\in S_n}q^{\ell (\sigma)}= \prod_{1\leq j < k\leq n}\frac{1-q^{1+k-j}}{1-q^{k-j}} = \prod_{1\leq j \leq n}\frac{1-q^j}{1-q}=[n]!
\end{equation}
(where $\ell(\sigma)$ denotes the {\em length} of $\sigma$).
Since the stabilizer subgroup $S_{n,\lambda}:=\{ \sigma\in S_n\mid \sigma\lambda=\lambda\}$
is isomorphic to the direct product $\prod_{l\in\mathbb{Z}} S_{m_l(\lambda)}$, the corresponding Poincar\'e polynomial
 factorizes in turn as
\begin{equation}\label{poincare-stabilizer}
S_{n,\lambda}(q)=\sum_{\substack{\sigma\in S_n\\ \sigma\lambda =\lambda}}q^{\ell (\sigma)}= \prod_{\substack{1\leq j < k\leq n\\ \lambda_j=\lambda_k}}\frac{1-q^{1+k-j}}{1-q^{k-j}}=\prod_{l\in\mathbb{Z}}[m_l(\lambda )]!
\end{equation}
and similarly
$S_{n,\lambda}\cap S_{n,\lambda+e_J}\cong \prod_{l\in\mathbb{Z}} (S_{m_{l,J}(\lambda)}\times S_{m_{l,J^c}(\lambda)})$ so
\begin{eqnarray}\label{poincare-intersection}
\lefteqn{(S_{n,\lambda}\cap S_{n,\lambda+e_J})(q)=} && \\
&& \prod_{\substack{j,k\in J\\ j<k,\, \lambda_j=\lambda_k}}\frac{1-q^{1+k-j}}{1-q^{k-j}}
\prod_{\substack{j,k\in J^c\\ j<k,\, \lambda_j=\lambda_k}}\frac{1-q^{1+k-j}}{1-q^{k-j}}
=\prod_{l\in\mathbb{Z}}[m_{l,J}(\lambda )]! [m_{l,J^c}(\lambda )]! ,\nonumber
\end{eqnarray}
\end{subequations}
where $m_{l, J}(\lambda)$ denotes the number of components $\lambda_j$, $j\in J$ such that $\lambda_j=l$ (i.e. $m_{l,J}(\lambda)+m_{l,J^c}(\lambda)=m_l(\lambda)$). Division of Eqs \eqref{poincare-stabilizer} and \eqref{poincare-intersection} now reveals that
\begin{equation}\label{binom}
V_{\lambda,J}=V_{\lambda,J}V_{\lambda,J^c}=
\frac{S_{n,\lambda}(q)}{(S_{n,\lambda}\cap S_{n,\lambda+e_J})(q)}=
 \prod_{l\in \mathbb{Z}}  { m_l(\lambda) \brack m_{l,J}(\lambda) }
\end{equation}
where ${m\brack k}:= \frac{[m]!}{[k]! [m-k]!}$ for $m\geq k\geq 0$. (Notice in this connection that here
$V_{\lambda,J^c}=1$, because the
product
in question is {\em empty}
as consequence of the assumption that $\lambda+e_J$ belongs to $\Lambda_n$.)

\subsection{}\label{step2}
It is clear by induction on $m\geq1$ that for any $f\in\mathcal{F}_n$ and $\lambda\in\Lambda_n$:
\begin{equation}\label{creation}
((a_l^*)^m f)( \lambda)=
\begin{cases}
[m]! { m_l(\lambda)\brack m}
 f(a_l^m\lambda ) & \text{if}\ m\leq m_l(\lambda),\\
0&\text{if}\ m> m_l(\lambda) ,
 \end{cases}
\end{equation}
where in the former case $a_l^m\lambda=\lambda+e_d+e_{d+1}+\cdots +e_{d+m-1}$---with $d=d(\lambda ,l):=\min\{ j \mid \lambda_j=l\}$---belongs to $\Lambda_n$ (because of the condition that $m\leq m_l(\lambda)$).
By iterating the formula in Eq.~\eqref{creation} it readily follows that---for
$l_1>l_2>\cdots >l_p$ and a composition $(m_1,m_2,\ldots,m_p)$ obtained by reordering the parts of the partition
$\eta=(\eta_1,\eta_2, \ldots ,\eta_p)$ (with $\eta_p\geq 1$)---the action of the corresponding monomial in $m_\eta (a^*)$ \eqref{monomial} is given by
\begin{align}
 ((a^*_{l_1})^{m_1} \cdots (a^*_{l_p})^{m_p} f)(\lambda)&=
 [\eta]!
 { m_{l_1}(\lambda)\brack m_1}\cdots { m_{l_p}(\lambda)\brack m_p} f(a_{l_p}^{m_p}\cdots a_{l_1}^{m_1}\lambda) \nonumber\\
&= [\eta]! \prod_{l\in\mathbb{Z}} { m_l(\lambda)\brack m_{l,J}(\lambda)}f(\lambda +e_J)  \label{action-term}
\end{align}
provided $m_k\leq m_{l_k}(\lambda)$ for $k=1,\ldots ,p$, and equal to zero otherwise. Here
\begin{equation}
J=\{ d_k,d_k+1,\ldots ,d_k+m_k-1\mid k=1,\ldots ,p\}
\end{equation}
with $d_k:=d(\lambda,l_k)=\min\{ j \mid \lambda_j=l_k\}$. The condition
that $m_k\leq m_{l_k}(\lambda)$ for $k=1,\ldots ,p$ guarantees that $J$ is a subset of
 $\{ 1,\ldots ,n\}$ of cardinality $|J|=m_1+\cdots +m_p=|\eta|$
and that
$\lambda+e_J=a_{l_p}^{m_p}\cdots a_{l_1}^{m_1}\lambda\in\Lambda_n$.

\subsection{}\label{step3}
From Steps \ref{step1} and \ref{step2}  one learns  that for $r\leq n$ the action
$(H_r^*f)(\lambda)$  is built of a sum of terms of the form $V_{\lambda,J}f(\lambda+e_J)$ (cf. Eqs. \eqref{binom}  and \eqref{action-term}), with $J\subset\{ 1,\ldots ,n\}$ satisfying that $|J|=r$ and $\lambda+e_J\in\Lambda_n$. For
$r>n$ on the other hand $(H_r^*f)(\lambda)$ vanishes (since then for all monomial terms
$m_1+\cdots +m_p=|\eta|=r>n\geq m_{l_1}(\lambda) +\cdots +m_{l_p}(\lambda)$).

To complete the proof of the explicit formula for $H_r^*f$ in Theorem \ref{Hr-action:thm},
it only remains to infer that in the former situation all terms on the RHS actually do occur and with multiplicity $1$.
Indeed, this is clear from the observation that
given $\lambda\in\Lambda_n$ and $J\subset \{ 1,\ldots ,n\}$ such that $\lambda+e_J\in\Lambda_n$, the corresponding $l_1>l_2>\cdots >l_p$ and $m_1,m_2,\ldots ,m_p$ for which
$$([m_1]! \cdots [m_p]!)^{-1}((a^*_{l_1})^{m_1} \cdots (a^*_{l_p})^{m_p} f)(\lambda)=V_{\lambda,J}f(\lambda+e_J)$$
are uniquely
 retrieved
by ordering the elements of the set $\{ \lambda_j\mid j\in J\}=\{ l_1,\ldots ,l_p\}$ and picking
$m_k=m_{l_k ,J}(\lambda)$, $k=1,\ldots ,p$.

\section{Diagonalization}
For $\lambda\in \Lambda_n$ and a spectral parameter $\xi=(\xi_1,\xi_2,\ldots ,\xi_n)$ taken from the open fundamental alcove
\begin{equation}\label{alcove}
A:=\{\xi\in\mathbb{R}^n\mid \pi >\xi_1>\xi_2>\cdots >\xi_n>-\pi\} ,
\end{equation}
let us define the $n$-variable Hall-Littlewood function as \cite[Ch.~III]{mac:symmetric}
\begin{subequations}
\begin{equation}\label{HLf}
\phi_\xi (\lambda) :=
\sum_{\sigma\in S_n}
C(\xi_{\sigma}) e^{i \lambda\cdot \xi_{\sigma}} ,
\end{equation}
with $\xi_\sigma:=(\xi_{\sigma_1},\xi_{\sigma_2},\ldots ,\xi_{\sigma_n})$
and
\begin{equation}\label{Cf}
C(\xi) :=\prod_{1\leq j<k \leq n} \frac{1-q e^{i(\xi_{k}-\xi_j)}}{1-e^{i(\xi_{k}-\xi_j)}}  .
\end{equation}
\end{subequations}
It is immediate from the explicit action in Theorem \ref{Hr-action:thm} and the Pieri formulas for the Hall-Littlewood functions \cite[Sec.~III.3]{mac:symmetric} that the action of the commuting operators $H_r$ and $H_r^*$
\eqref{integrals} in the $n$-particle subspace $\mathcal{F}(\Lambda_n)$  is diagonal on $\phi_\xi$ \eqref{HLf}, \eqref{Cf}.

\begin{corollary}[Diagonalization]\label{diagonal:cor} For any spectral value $\xi$ in the fundamental alcove $A$ \eqref{alcove},
the $n$-variable Hall-Littlewood function $\phi_\xi$ \eqref{HLf}, \eqref{Cf} constitutes  a joint eigenfunction for the commuting operators $H_r$ and $H_r^*$ \eqref{integrals} in $\mathcal{F}(\Lambda_n)$:
\begin{subequations}
\begin{equation}\label{ev:eq}
H_r\phi_\xi =  e_r(e^{-i\xi}) \phi_\xi \quad \text{and}\quad
H_r^*\phi_\xi =  e_r(e^{i\xi}) \phi_\xi  \qquad (r=1,\ldots ,n),
\end{equation}
with $  e_r(e^{-i\xi}):=e_r(e^{-i\xi_1},\ldots ,e^{-i\xi_n})$ and $e_r(e^{i\xi}):=e_r(e^{i\xi_1},\ldots ,e^{i\xi_n})$, where $e_r$ refers to the $r$th elementary symmetric function
\begin{equation}\label{erf}
e_r(x_1,\ldots ,x_n):=\sum_{1\leq j_1<j_2<\cdots <j_r\leq n} x_{j_1}x_{j_2}\cdots x_{j_r}.
\end{equation}
\end{subequations}
\end{corollary}
\begin{proof}
By Theorem \ref{Hr-action:thm}, the eigenvalue equations in Eqs. \eqref{ev:eq}, \eqref{erf} become explicitly:
 \begin{equation*}
    \begin{split}
e_r(e^{-i\xi}) \phi_\xi(\lambda)  &= \sum_{\substack{J\subset \{1,2,\ldots,n\} ,|J|=r\\ \lambda-e_J\in\Lambda_n  }} V_{\lambda ,J^c} \phi_\xi (\lambda-e_J) ,\\
e_r(e^{i\xi}) \phi_\xi(\lambda )  &= \sum_{\substack{J\subset \{1,2,\ldots,n\} ,|J|=r\\ \lambda+e_J\in\Lambda_n  }}
V_{\lambda ,J} \phi_\xi (\lambda+e_J) ,
  \end{split}
 \end{equation*}
 respectively.
 Both formulas boil down to well-known Pieri identities for the Hall-Littlewood functions \cite[Sec.~III.3]{mac:symmetric}. In the form stated above
the second identity can e.g. be directly retrieved from  \cite[Eq.~(C.11)]{die-ems:unitary} and the first identity follows for $r<n$ from the second upon dividing by $e_n(e^{i\xi})$ and replacing $r$ by $n-r$, whereas for $r=n$
both identities are equivalent and reduce to the elementary translational quasi-periodicity
$\phi_\xi (\lambda+e_1+\cdots+e_n)=e^{i\xi_1+\cdots+i\xi_n}\phi_\xi (\lambda)$ (which is manifest from Eq. \eqref{HLf}).
\end{proof}

It follows in particular that the Hall-Littlewood function $\phi_\xi$ \eqref{HLf}, \eqref{Cf} is  an eigenfunction of the $q$-boson Hamiltonian $\text{H}_q$ \eqref{qbH}, \eqref{hopping-op} in the $n$-particle subspace $\mathcal{F}(\Lambda_n)$:
\begin{equation}
\text{H}_q\phi_\xi=\varepsilon (\xi) \phi_\xi \quad\text{with}\quad \varepsilon(\xi):=2\sum_{j=1}^n \cos (\xi_j)
\end{equation}
(cf Eq. \eqref{hamiltonian}).

\section{Spectral analysis}
To address the completeness of the above eigenfunctions for the infinite $q$-boson system, we pass from our algebraic Fock space $\mathcal{F}$ \eqref{AFock} to a full-fledged Fock space
\begin{equation}\label{HFock}
\mathcal{H}:=\bigoplus_{n\geq 0} \ell^2(\Lambda_n,\delta_n),
\end{equation}
which is built of all linear combinations $\sum_{n\geq 0} c_n f_n$---with $c_n\in\mathbb{C}$ and $f_n\in \ell^2(\Lambda_n,\delta_n)$---such that $\sum_{n\geq 0} |c_n|^2 \langle f_n,f_n\rangle_n<\infty$. Here the $n$-particle Hilbert space $\ell^2(\Lambda_n,\delta_n)=\mathcal{H}\cap\mathcal{F}(\Lambda_n)$ consists of the functions $ f\in\mathcal{F}(\Lambda_n)$ such that
$ \langle f ,f\rangle_n <\infty$, where
\begin{subequations}
\begin{equation}\label{ipn}
\langle f,g\rangle_n:=\sum_{\lambda\in\Lambda_n} f(\lambda) \overline{g(\lambda)} \delta_n(\lambda)
\qquad (f,g\in \ell^2(\Lambda_n,\delta_n))
\end{equation}
with
\begin{equation}
\delta_n(\lambda):=1/S_{n,\lambda}(q) =1/\prod_{l\in\mathbb{Z}} [m_l(\lambda )]!
\end{equation}
\end{subequations}
(cf Eq. \eqref{poincare-stabilizer}).

The representation of the $q$-boson algebra in Eq. \eqref{qboson-rep} readily extends from the dense domain
\begin{equation}\label{domain}
\mathcal{D}:=\mathcal{H}\cap\mathcal{F}
\end{equation}
consisting of the  {\em finite}  linear combinations  $\sum_{n\geq 0} c_n f_n$---with $c_n\in\mathbb{C}$ and $f_n\in \ell^2(\Lambda_n,\delta_n)$---to a bounded representation on the Fock space $\mathcal{H}$ \eqref{HFock}. Indeed,  it is immediate from the definitions that for any
$f\in  \ell^2(\Lambda_{n},\delta_{n})$:
\begin{align}
\langle \beta_l f,\beta_l f\rangle_{n-1}&\leq  (1-q)^{-1} \langle  f,f\rangle_{n}  ,\nonumber\\
\langle \beta_l^* f,\beta_l^*f\rangle_{n+1}&\leq  (1-q)^{-1} \langle  f,f\rangle_{n},\\
\langle N_l f,N_lf\rangle_{n}&\leq  \langle  f,f\rangle_{n}  \nonumber
\end{align}
(where one exploits that $\delta_{n+1,\beta_l^*\lambda}(q)=\delta_{n,\lambda}(q)/[m_l(\lambda)+1]$ for all $\lambda\in\Lambda_n$ and that $[m]\leq 1/(1-q)$ for all $m=0,1,2,\ldots$). The representation at issue moreover preserves the $*$-structure:
\begin{equation}\label{ca-adjointness}
\langle \beta_l^*f,g\rangle_{n+1}= \langle f, \beta_l g\rangle_{n}\quad\text{and}\quad
\langle N_lf,g\rangle_{n}= \langle f, N_l g\rangle_{n}
\end{equation}
(for all $f\in  \ell^2(\Lambda_{n},\delta_{n})$ with
$g\in  \ell^2(\Lambda_{n+1},\delta_{n+1}) $ and $g\in  \ell^2(\Lambda_{n},\delta_{n})$, respectively).

The completeness of the eigenfunctions in Corollary \ref{diagonal:cor} is now obvious from the well-known fact that the Hall-Littlewood functions $\phi_\xi(\lambda)$, $\lambda\in\Lambda_n$
form an orthogonal basis for the Hilbert space $L^2(A,\Delta \text{d}\xi)$ with inner product
\begin{equation}\label{orthogonality}
\langle \hat{f},\hat{g}\rangle_\Delta=\frac{1}{(2\pi)^n}\int_A \hat{f}(\xi)\overline{\hat{g}(\xi)}\Delta(\xi)\text{d}\xi ,
\quad \text{where}\quad \Delta (\xi):= \frac{1}{|C(\xi)|^2}
\end{equation}
with $C(\xi)$ taken from Eq. \eqref{Cf}. More specifically, for any $\lambda, \mu\in\Lambda_n$
one has that \cite[\S 10]{mac:orthogonal}
\begin{equation}
\langle \phi(\lambda) ,\phi (\mu) \rangle_\Delta =\begin{cases}
1/\delta_n(\lambda)&\text{if}\ \lambda =\mu ,\\
0&\text{otherwise}.
\end{cases}
\end{equation}
The corresponding Fourier transform $\boldsymbol{F_q}:\ell^2(\Lambda_n,\delta_n)\to L^2(A,\Delta\text{d}\xi)$
defined by
\begin{subequations}
\begin{equation}\label{ft1}
(\boldsymbol{F_q}f)(\xi):= \langle f,\phi_\xi \rangle_n=\sum_{\lambda\in\Lambda_n}f(\lambda)
\overline{\phi_\xi (\lambda)}\delta_n(\lambda)
\end{equation}
($f\in \ell^2(\Lambda_n,\delta_n)$) thus determines a Hilbert space isomorphism with the inversion formula given by
\begin{equation}\label{ft2}
(\boldsymbol{F_q}^{-1}\hat{f})(\lambda) = \langle \hat{f},\overline{\phi ( \lambda)}\rangle_\Delta=
\frac{1}{(2\pi)^n}\int_A \hat{f}(\xi) \phi_\xi (\lambda)\Delta(\xi)\text{d}\xi
\end{equation}
\end{subequations}
$(\hat{f}\in L^2(A,\Delta\text{d}\xi))$.

By Corollary \ref{diagonal:cor}, this means that in the $n$-particle subspace $\ell^2(\Lambda_n,\delta_n)$ the higher commuting $q$-boson Hamiltonians
\begin{equation}\label{hamiltonians}
\text{H}_{q,r}:= H_r+H_r^*,\qquad r=1,\ldots ,n,
\end{equation}
are unitarily equivalent to bounded self-adjoint multiplication operators $\hat{\text{E}}_1,\ldots ,\hat{\text{E}}_n$ on
$L^2(A,\Delta\text{d}\xi)$ of the form
\begin{subequations}
\begin{equation}
(\hat{\text{E}}_r\hat{f})(\xi):=\varepsilon_r(\xi)\hat{f}(\xi)
\end{equation}
with
 \begin{equation}\label{ev-r}
\varepsilon_r(\xi):= 2\sum_{1\leq j_1<j_2<\cdots <j_r\leq n} \cos (\xi_{j_1}+\xi_{j_2}+\cdots +\xi_{j_r}),
\end{equation}
\end{subequations}
viz.
\begin{equation}\label{s-d}
\text{H}_{q,r}=\boldsymbol{F_q}^{-1}  \circ\hat{\text{E}}_r \circ\boldsymbol{F_q},\qquad r=1,\ldots,n,
 \end{equation}
 on $\ell^2(\Lambda_n,\delta_n)$.
 \begin{theorem}[Spectral decomposition in $\ell^2(\Lambda_n,\delta_n)$]\label{hamiltonians:thm}
 The higher $q$-boson Hamiltonians $\text{H}_{q,1},\ldots ,\text{H}_{q,n}$ \eqref{hamiltonians} consitute $n$ independent commuting bounded self-adjoint operators on $\ell^2(\Lambda_n,\delta_n)$ with purely absolutely continuous spectrum. The spectral decomposition of these Hamiltonians in the $n$-particle subspace $\ell^2(\Lambda_n,\delta_n)$ is given explicitly by Eq. \eqref{s-d}.
 \end{theorem}

As a consequence, the infinite  $q$-boson hierarchy
\begin{equation}\label{hierarchy}
\text{H}_{q,r}= H_r+H_r^*,\qquad r\in\mathbb{N}
\end{equation}
consists in turn of (commuting) symmetric operators  in Fock space on the dense domain $\mathcal{D}$ \eqref{domain}.
The operators in question turn out to be essentially self-adjoint  and unbounded in $\mathcal{H}$ \eqref{HFock}, because
for $z\in\mathbb{C}\setminus\mathbb{R}$ the range $(\text{H}_{q,r}-z)\mathcal{D}$ is dense in $\mathcal{H}$
(as $(\text{H}_{q,r}-z)$ maps $\ell^2(\Lambda_n,\delta_n)$  onto itself)
and $\lim_{n\to\infty} \sup_{\xi\in A} |\varepsilon_r(\xi)|=\infty$, respectively. This permits to adapt the integrability result in Corollary \ref{integrability:cor} to the present setting as follows.

\begin{theorem}[Integrability] The higher Hamiltonians $\text{H}_{q,r}$  \eqref{hierarchy} of the infinite $q$-boson hierarchy on $\mathcal{D}$ \eqref{domain} extend uniquely to independent unbounded
self-adjoint operators in the Fock space $\mathcal{H}$ \eqref{HFock} with commuting resolvents $(\text{H}_{q,r}-z)^{-1}$ (with  $z\in\mathbb{C}\setminus\mathbb{R}$ and $r\in\mathbb{N}$).
\end{theorem}

\begin{remark}\label{Hr-action:rem}
It is a priori clear already from the definition of $H_r$ and $H_r^*$ \eqref{integrals} in terms of hopping operators---without need to resort to the explicit formula in Theorem \ref{Hr-action:thm}---that for given $\lambda\in\Lambda_n$ and any $f\in\mathcal{F}(\Lambda_n)$ the values of
$(H_rf)(\lambda)$ and $(H_r^*f)(\lambda)$ involve only evaluations of $f$ at a {\em finite} number of points in
$\Lambda_n$. Hence, for
$f$ in the subspace $C_0(\Lambda_n)\subset\mathcal{F}(\Lambda_n)\cap \ell^2(\Lambda_n,\delta_n)$ of functions with {finite} support in $\Lambda_n$ the infinite sums comprising $H_rf$ and $H_r^*f$ contain only a finite number of nonvanishing monomial terms. The symmetry
\begin{equation}
\langle H_rf,g\rangle_n=\langle f,H_r^*g\rangle_n\qquad (\forall f,g \in C_0(\Lambda_n))
\end{equation}
then follows from the first relation in Eq. \eqref{ca-adjointness} (without invoking the spectral decomposition in Eq. \eqref{s-d}).
One can thus determine the action of $H_r$ on $C_0(\Lambda_n)$ from the action of $H_r^*$ (and vice versa) by computing the adjoint with respect to the inner product $\langle \cdot ,\cdot\rangle_n$:
\begin{align*}
\langle H_rf,g\rangle_n&=\langle f,H_r^*g\rangle_n=\sum_{\lambda\in\Lambda_n} \delta_n(\lambda) f(\lambda)\overline{(H_r^*g)(\lambda)}\\
&=\sum_{\lambda\in\Lambda_n} \delta_n(\lambda) f(\lambda)
\sum_{\substack{J\subset \{1,\ldots,n\} ,|J|=r\\ \lambda+e_J\in\Lambda_n }} V_{\lambda,J}
\overline{g(\lambda+e_J)} \\
&=\sum_{\lambda\in\Lambda_n} \delta_n(\lambda)\overline{g(\lambda)}
\sum_{\substack{J\subset \{1,\ldots,n\} ,|J|=r\\ \lambda-e_J\in\Lambda_n }} V_{\lambda,J^c}
f(\lambda-e_J),
\end{align*}
where it was used in the last step that $\delta_n(\lambda-e_J)V_{\lambda-e_J,J}=\delta_n(\lambda)V_{\lambda ,J^c}$ (for $\lambda\in\Lambda_n$ such that $\lambda-e_J\in\Lambda_n$).
Since the actions of $H_r$ and $H_r^*$ on $\mathcal{F}(\Lambda_n)$  are determined completely by their restrictions to the subspace $C_0(\Lambda_n)$ (by the opening statement of this remark), the above computation shows that both formulas in Theorem \ref{Hr-action:thm} follow from each other (so it indeed suffices in the proof of Theorem \ref{Hr-action:thm} to verify only one of these two cases directly).
\end{remark}

\section{$n$-Particle scattering}
For $q\to 0$, the Hall-Littlewood functions \eqref{HLf} , \eqref{Cf}
reduce to Schur functions \cite[Ch.~I]{mac:symmetric}; the $q$-boson system degenerates in this limit to a system of impenetrable bosons known as the {\em Phase Model} \cite{bog-ize-kit:correlation,tsi:quantum,kor:cylindric}. We end up by computing the scattering operator that compares the large-time asymptotics of the $n$-particle dynamics of the $q$-boson system with that of the phase model.

It is manifest from the explicit product formula for the orthogonality measure $\Delta (\xi )$ \eqref{orthogonality} that the restrictions of the Hamiltonians $\text{H}_{q,1},\ldots ,\text{H}_{q,n}$ \eqref{hamiltonians} to the $n$-particle subspace  $\ell ^2(\Lambda_n,\delta_n)$ fit within a much larger class of discrete integrable lattice systems on the discrete cone $\Lambda_n$ \eqref{dominant} for which the scattering behavior was analyzed in great detail by  Ruijsenaars \cite{rui:factorized}.
Upon identifying how Ruijsenaars' general results specialize to the case of the infinite $q$-boson model, the desired scattering operator follows immediately.
To this end it is convenient to pass to uniform Lebesgue measures by incorporating orthogonality densities into the wave functions via the following gauge transformation:
\begin{subequations}
\begin{equation}
\begin{split}
 \Psi_\xi(\lambda)&:=i^{n(n-1)/2} \Delta(\xi)^{1/2} \delta_n(\lambda )^{1/2} \phi_\xi (\lambda ) \\
 &=\delta_n(\lambda )^{1/2}\sum_{\sigma \in S_n}\text{sign} (\sigma )
\hat{\mathcal S}_\sigma(\xi)^{1/2} e^{i(\rho+\lambda)\cdot\xi_\sigma}
 \end{split}
\end{equation}
($\xi\in A$ \eqref{alcove}),  where $\rho:=\frac{1}{2}(n-1,n-3,n-5,\ldots ,3-n,1-n)$ and
\begin{equation}
\hat{ {\mathcal S}}_{\sigma} (\xi)
 := \prod_{\substack{1\leq j<k\leq n\\ \sigma^{-1}_j<\sigma^{-1}_k}} s(\xi_k-\xi_j)
      \prod_{\substack{1\leq j<k\leq n\\ \sigma^{-1}_j>\sigma^{-1}_k}} \overline{ s(\xi_k-\xi_j)} ,
\end{equation}
with $\sigma_j^{-1}:=(\sigma^{-1})_j$ for $j=1,\ldots ,n$ and
\begin{equation}
 s(x)^{1/2}:=\frac{1-qe^{ix}}{|1-q e^{ix}|}, \qquad\text{so}\quad s(x)=\frac{1-qe^{ix}}{1-qe^{-ix}}.
\end{equation}
\end{subequations}
The wave functions in question diagonalize the commuting self-adjoint difference operators
\begin{equation}\label{Hqr-tilde}
\tilde{\text{H}}_{q,r}:=\delta_n^{1/2}\text{H}_{q,r}\delta_n^{-1/2},\qquad r=1,\ldots,n,
\end{equation}
in $\ell^2(\Lambda_n)$, viz.
\begin{equation}
\tilde{\text{H}}_{q,r}=\boldsymbol{\tilde{F}_q}^{-1}  \circ \hat{\text{E}}_r \circ \boldsymbol{\tilde{F}_q},\qquad r=1,\ldots,n
 \end{equation}
(cf. Eq. \eqref{s-d}), where $\boldsymbol{\tilde{F}_q}:\ell^2(\Lambda_n)\to L^2(A,\text{d}\xi)$ denotes the Hilbert space isomorphism defined by
\begin{subequations}
\begin{equation}
(\boldsymbol{\tilde{F}_q}f)(\xi):= \sum_{\lambda\in\Lambda_n}f(\lambda)\overline{\Psi_\xi(\lambda)}
\qquad (f\in \ell^2(\Lambda_n))
\end{equation}
with
\begin{equation}
(\boldsymbol{\tilde{F}_q}^{-1}\hat{f})(\lambda) =
\frac{1}{(2\pi)^n}\int_A \hat{f}(\xi) \Psi_\xi(\lambda)\text{d}\xi \qquad (\hat{f}\in L^2(A,\text{d}\xi))
\end{equation}
\end{subequations}
(cf. Eqs. \eqref{ft1}, \eqref{ft2}). The action of $\tilde{\text{H}}_{q,r} $ \eqref{Hqr-tilde} on $f\in\ell^2(\Lambda_n)$ reads explicitly
\begin{equation}
\begin{split}
(\tilde{\text{H}}_{q,r} f)(\lambda) =
&\sum_{\substack{J\subset\{1,\ldots,n\} ,|J|=r\\ \lambda+e_J\in\Lambda_n  }}
V_{J,\lambda}^{1/2} V_{J^c,\lambda+e_J}^{1/2}
f(\lambda+e_J)\ +\\
&\sum_{\substack{J\subset\{1,\ldots,n\} ,|J|=r\\ \lambda-e_J\in\Lambda_n  }}
V_{J^c,\lambda}^{1/2} V_{J,\lambda-e_J}^{1/2}
f(\lambda-e_J) .
\end{split}
\end{equation}

For $1\leq r\leq n$, let $A_r$ be an open dense domain in $A$ \eqref{alcove} on which the gradient vector $\nabla \varepsilon_r$ is regular with respect to the permutation-action of $S_n$ on its components:
\begin{equation}\label{regular-domain}
 A_r:=\{  \xi \in A \mid \partial_j\varepsilon_r \neq \partial_k\varepsilon_r,\,\forall 1\leq j<k\leq n\}
\end{equation}
(with $\varepsilon_r$ taken from Eq. \eqref{ev-r}).
For any $\xi \in A_r$, there exists then a unique permutation $\sigma_\xi\in S_n$ reordering the components
of $\nabla \varepsilon_r(\xi)$ in strictly decreasing order, i.e.
$\sigma_\xi(\nabla \varepsilon_r(\xi))\in\mathbb{R}^n_>:=\{ x\in\mathbb{R}^n\mid x_1>x_2>\cdots >x_n\}$. Clearly the assignment $\xi\to\sigma_\xi$ is constant on the connected components of $A_r$ by the continuity of $\nabla \varepsilon_r(\xi)$.
Let $\hat{\mathcal{S}}_r$ now denote the following unitary operator
on $L^2(A,\text{d}\xi)$---the {\em scattering matrix}---defined via its restriction to the dense subspace of smooth test functions with compact support inside $A_r$:
\begin{equation}
(\hat{\mathcal S}_r \hat f)(\xi) = \hat{\mathcal S}_{\sigma_\xi}(\xi) \hat f(\xi)\quad (\hat f\in C_0^\infty(A_r)).
\end{equation}
The following scattering theorem---providing explicit wave- and scattering operators that compare the large-time asymptotics of the dynamics
\begin{equation}\label{dynamics}
(e^{it\tilde{\text{H}}_{q,r} }f)(\lambda)=
\frac{1}{(2\pi)^n}\int_A e^{it\varepsilon_r (\xi )}\hat{f}(\xi) \Psi_\xi (\lambda)\text{d}\xi\qquad
(\hat{f}= \boldsymbol{\tilde{F}_q} f)
\end{equation}
of the higher $q$-boson Hamiltonian $\tilde{\text{H}}_{q,r}$ \eqref{Hqr-tilde}
with that of the corresponding $q\to 0$ limiting Hamiltonian $\tilde{\text{H}}_{0,r}$:
\begin{equation}
(\tilde{\text{H}}_{0,r} f)(\lambda) =
\sum_{\substack{J\subset\{1,\ldots,n\} ,|J|=r\\ \lambda+e_J\in\Lambda_n  }} f(\lambda+e_J) +\\
\sum_{\substack{J\subset\{1,\ldots,n\} ,|J|=r\\ \lambda-e_J\in\Lambda_n  }} f(\lambda-e_J)
\end{equation}
$(f\in \ell^2(\Lambda_n))$ for the phase model of impenetrable bosons---is a very special case of \cite[Thm.~3.3]{rui:factorized}.

\begin{theorem}[Wave and scattering operators]\label{scattering:thm}
  The operator limits
\begin{subequations}
\begin{equation}
\Omega^{\pm}_r :=s-\lim_{t\to \pm \infty}  e^{i t  \tilde{\text{H}}_{q,r}}e^{-it \tilde{\text{H}}_{0,r}}
\end{equation}
converge in the strong $\ell^2(\Lambda_n)$-norm topology and the corresponding wave operators $\Omega^\pm_r$ are given by unitary operators in $\ell^2(\Lambda_n)$ of the form
\begin{equation}
\Omega_r^\pm = \boldsymbol{\tilde{F}_q}^{-1} \circ \hat{\mathcal S}_r^{\mp 1/2}  \circ \boldsymbol{\tilde{F}_0}.
\end{equation}
Consequently, the scattering operator comparing the dynamics of $\tilde{\text{H}}_{q,r}$  and $\tilde{\text{H}}_{0,r}$ is given by the unitary operator
\begin{equation}
\mathcal{S}_r:=(\Omega_r^+)^{-1} \Omega_r^- =  \boldsymbol{\tilde{F}_0}^{-1}  \circ \hat{\mathcal S}_r \circ  \boldsymbol{\tilde{F}_0}
\end{equation}
\end{subequations}
($r=1,\ldots ,n$).
\end{theorem}

It is instructive to outline briefly how Ruijsenaars' proof in \cite{rui:factorized} simplifies in our particular situation.
For this purpose,
we associate to any $\hat f\in C_0^\infty(A_r)$ a \emph{($q=0$) boson wave packet} $f^{(0)}(t)$  and  \emph{$q$-boson wave packets} $f_\pm(t)$ in $\ell^2(\Lambda_n)$ of the form
\begin{equation}
  \begin{split}
 f^{(0)}(t)&:=  \boldsymbol{\tilde{F}_0}^{-1}  (e^{-it \hat{E}_r} \hat{f}),\\
 f_\pm(t)&:= \boldsymbol{\tilde{F}_q}^{-1} ( e^{-it \hat{E}_r} \hat{\mathcal S}_r^{\pm 1/2}\hat{f}) .
  \end{split}
\end{equation}
Theorem \ref{scattering:thm} is now immediate from the following proposition.
\begin{proposition}[Asymptotic equivalence] \label{af:prp}
For all $K>0$ one has that
\begin{equation}
||  f_\pm(t) -f^{(0)}(t) ||  = O(|t|^{-K}) \quad \text{as}\ t\to \pm\infty
\end{equation}
(where $\|\cdot\|$ refers to the $\ell^2$-norm in $\ell^2(\Lambda_n)$).
\end{proposition}

To infer the above proposition, one may assume without loss of generality that the compact support of smooth test function $\hat f$ is contained inside a connected component of $A_r$. We then write $\hat{\sigma}\in S_n$ for the unique ($\xi$-independent) permutation ordering the elements of  $\nabla \varepsilon_r(\xi)$ in strictly decreasing order for all $\xi$ in the support of $\hat{f}$.
Let $V_{\text{class}}\subset \mathbb{R}^n$  be an open bounded neighborhood of the compact range of classical wave-packet velocities $\text{Ran}(\nabla \varepsilon_r):=\{ \nabla \varepsilon_r(\xi)\mid \xi\in\text{Supp}(\hat f)\}$ staying away from the boundary of the chamber
$\hat{\sigma}^{-1}(\mathbb{R}^n_>)$.  The \emph{classical wave packet}, finitely supported on the following $t$-dependent region of $\Lambda_n$:
\begin{equation}
\Lambda_n^{clas}(t)
:=
\begin{cases}
  \{ \rho+\lambda\in t \hat{\sigma} (V_{clas}) \} & \text{for}\ t>0,\\
  \{ \rho+\lambda\in t \sigma_0  \hat{\sigma} (V_{clas}) \} & \text{for}\ t< 0,
\end{cases}
\end{equation}
is defined as
\begin{multline}
f_\lambda^{clas}(t)   :=
\\  \begin{cases}
    \frac{\text{sign}(\hat{\sigma})}{(2\pi)^n}\int_{A} e^{i(\rho+\lambda)\cdot \xi -it \varepsilon_r(\xi)}\hat f(\xi) \text{d}\xi
              & \text{for}\ \lambda\in \Lambda_n^{clas}(t)\ \text{and}\ t>0,\\
    \frac{\text{sign}(\hat{\sigma})}{(-2\pi)^n}\int_{ A} e^{i(\rho+\lambda)\cdot \xi_{\sigma_0} -it \varepsilon_r(\xi)}\hat f(\xi) \text{d}\xi
              & \text{for}\ \lambda\in \Lambda_n^{clas}(t)\ \text{and}\ t<0,\\
    0          & \text{otherwise.}
  \end{cases}
\end{multline}
Here $\sigma_0$ refers to the order reversing permutation for which $\sigma_j=n+1-j$, $j=1,\ldots ,n$

With the aid of the following stationary phase estimate from \cite[p. 38-39]{ree-sim:methods}: for any $K>0$ there exists a constant $C_K>0$ such that
\begin{equation}
\Big| \int_{\mathbf A} e^{ix\cdot \xi -it\varepsilon _r(\xi)}\hat f(\xi)\, \text{d}\xi\, \Big| \leq \frac{C_K}{(1+|x|+|t|)^K}
\end{equation}
for all $x\in \mathbb{R}^n$ and $t\in \mathbb{R}$ such that $x\not\in t V_{clas}$, it is now not difficult to deduce that
\begin{equation}\label{e:asym10}
  \begin{split}
  ||  f^{(0)}(t) - f^{clas}(t)  ||  &= O(|t|^{-K}) \quad \text{as}\ t\to \pm\infty ,\\
  ||  f_\pm(t) - f^{clas}(t)  ||  &= O(|t|^{-K}) \quad \text{as}\ t\to \pm\infty ,
  \end{split}
 \end{equation}
whence the asymptotic equivalence in Proposition \ref{af:prp} follows.

\section*{Acknowledgments} It is a pleasure to thank Edwin Langmann and Christian Korff
for helpful conversations and comments.

\bibliographystyle{amsplain}

\begin{thebibliography}{000000}

\bibitem[BIK]{bog-ize-kit:correlation} N.M. Bogoliubov, A.G. Izergin, and A.N. Kitanine,
Correlation functions for a strongly correlated boson system,
Nuclear Phys. B {\bf 516} (1998), 501--528.

\bibitem[D]{die:diagonalization} J.F. van Diejen,
Diagonalization of an integrable discretization of the repulsive delta Bose gas on the circle,
Comm. Math. Phys. {\bf 267} (2006), 451--476.

\bibitem[DE]{die-ems:unitary} J.F. van Diejen and E. Emsiz, Unitary representations of
affine Hecke algebras related to Macdonald spherical functions, J. Algebra {\bf 354} (2012), 180--210.

\bibitem[FW]{fod-whe:hall-littlewood}
O. Foda and M. Wheeler, Hall-Littlewood plane partitions and KP, Int. Math. Res. Not. IMRN {\bf 2009} (2009),
2597--2619.

\bibitem[J]{jin:vertex} N.H. Jing,
Vertex operators and Hall-Littlewood symmetric functions,
Adv. Math. {\bf 87} (1991), 226–-248.

\bibitem[KS]{kli-smu:quantum} A. Klimyk and K. Schm\"udgen, {\em Quantum groups and their representations}, Springer-Verlag, Berlin, 1997.

\bibitem[K]{kor:cylindric} C. Korff, Cylindric versions of specialised Macdonald functions and a deformed Verlinde algebra, Comm. Math. Phys. {\bf 318} (2013), 173--246.

\bibitem[LL]{lie-lin:exact} E.H. Lieb and W. Liniger,
Exact analysis of an interacting Bose gas. I. The general solution and the ground state,
Phys. Rev. (2) {\bf 130} (1963), 1605--1616.

\bibitem[M1]{mac:symmetric}  I.G. Macdonald, {\em Symmetric Functions and
Hall Polynomials}, Second Edition, Clarendon Press, Oxford, 1995.

\bibitem[M2]{mac:orthogonal}  \bysame, Orthogonal polynomials associated
with root systems, S\'em. Lothar. Combin. {\bf 45} (2000/01), Art.
B45a, 40 pp.

\bibitem[RS]{ree-sim:methods} M. Reed and B.  Simon, {\em Methods of Modern Mathematical Physics. III. Scattering Theory},  Academic Press, New York-London, 1979.

\bibitem[R1]{rui:complete} S.N.M. Ruijsenaars,
Complete integrability of relativistic Calogero-Moser systems and elliptic function identities, Comm. Math. Phys. {\bf 110} (1987), 191--213.

\bibitem[R2]{rui:factorized} \bysame, Factorized weight functions vs. factorized scattering,
Comm. Math. Phys. {\bf 228} (2002), 467--494.

\bibitem[T]{tsi:quantum} N.V. Tsilevich, The quantum inverse scattering method for the $q$-boson model and symmetric functions, Funct. Anal. Appl. {\bf 40} (2006), 207--217.

\end{thebibliography}

\end{document}